\documentclass[a4paper,12pt]{article}

\usepackage{tikz}
\usetikzlibrary {calc}
\usetikzlibrary {arrows.meta}
\usetikzlibrary {shapes.geometric}
\usetikzlibrary {graphs,shapes.geometric}
\usetikzlibrary {decorations.pathmorphing}

\usetikzlibrary {graphs}

\usepackage[shortlabels]{enumitem}
\usepackage[]{geometry}
\usepackage[utf8]{inputenc}
\usepackage{fancybox, calc}
\usepackage{authblk}
\usepackage{amssymb}
\usepackage{amsmath}
\usepackage{hyperref}
 \usepackage{systeme}
\usepackage{graphicx}
\usepackage{enumitem}
\usepackage{ucs}
\usepackage{amsfonts}
\usepackage{amssymb}
\usepackage{makeidx}
\usepackage{diagbox}
\usepackage{tabularx}
\usepackage{float}
\usepackage{url}
\usepackage{ marvosym }
\usepackage{ wasysym }
\usepackage{ mathrsfs }
\usepackage[utf8]{inputenc}
\usepackage[english]{babel}
\usepackage[normalem]{ulem}
\usepackage{ucs}
\usepackage{amsfonts}
\usepackage{dcolumn}
\usepackage{makeidx}
\usepackage{bm}
\usepackage{multirow}
\usepackage{graphicx}
\usepackage{hyperref}
\usepackage{tikz}
\usepackage{subfig}
\usepackage{booktabs}
\usepackage{xcolor}
\usepackage{fancyhdr}
\usepackage{amsmath, amsthm, amssymb}
\usepackage{multicol}
\usepackage{yfonts}
\usepackage{multicol}
\usepackage{caption}
\usepackage{cancel}
\usepackage{amssymb,amsfonts,amssymb,bbm}
\usepackage{amsmath}
\usepackage{phaistos,hieroglf,staves,ifsym,marvosym,skull}
\usepackage{tikz}
\usepackage{transparent}
\usepackage{eso-pic}
\usepackage{soul}
\usepackage{cite}
\usepackage{lettrine}

\definecolor{rougeG}{rgb}{.76,0,.12}
\definecolor{vertG}{rgb}{.07,.56,.25}

\numberwithin{equation}{section}
\def\Rset{\mathbb{R}}
\def\pdf{f}

\def\Rset{\mathbb{R}}

\def\pdf{f}

\newcommand{\descort}[2][]{\mathfrak{E}_{#1}\ifthenelse{\isempty{#2}}{}{ {\left[ #2 \right]}}}
\newcommand{\escort}[2][]{\varepsilon_{#1}\ifthenelse{\isempty{#2}}{}{ {\left[ #2 \right]}}}

\newtheorem{theorem}{Theorem}[section]

\numberwithin{equation}{section}
\restylefloat{table}

\title{New sharp inequalities involving non-relative, relative and cross informational functionals with some remarkable minimizers of generalized Gaussian and Beta types}
\author[1]{Razvan Gabriel Iagar\footnote{e-mail: razvan.iagar@urjc.es}}
\affil[1]{Departamento de Matemática Aplicada, Ciencia e Ingeniería de los Materiales y Tecnología Electrónica, Universidad Rey Juan Carlos,
		28933 Móstoles (Madrid), Spain}
\author[1,2]{David Puertas-Centeno\footnote{e-mail: david.puertas@urjc.es}}

\affil[2]{Data, Complex Networks and Cybersecurity Research Institute, Universidad Rey Juan Carlos, 28028 (Madrid), Spain}

\date{\today}

\begin{document}

\maketitle

\begin{abstract}
Several new and sharp informational inequalities are derived as a byproduct of Stam-like and moment-entropy-like inequalities in the relative framework and a recently established inequality mixing the R\'enyi entropy, the R\'enyi divergence and the R\'enyi cross entropy of suitable probability density functions. More precisely, we obtain a Stam-like inequality connecting the R\'enyi entropy power, the recently introduced scaling-invariant relative Fisher information and the R\'enyi cross entropy. Furthermore, we derive an inequality involving only Fisher-like informational measures and another inequality involving only moment-like functionals of non-relative, relative and cross types, respectively. All the inequalities are sharp. The minimizers of the Stam-like inequality are, in certain cases, pairs of Gaussian or stretched Gaussian probability densities; in contrast, each minimizer of the moment-like inequality is the probability density of the generalized Beta distribution.
\end{abstract}

\section{Introduction}

The derivation of sharp inequalities has been a cornerstone in the development of scientific knowledge. Fundamental inequalities such as the Cauchy--Schwarz, Jensen, H\"{o}lder or Gagliardo--Nirenberg inequalities have played a pivotal role across numerous areas of science, with strong implications not only in the theoretical advances but also in the understanding of physical phenomena. As a representative example, in statistical mechanics, the prominent role of the Gaussian probability density function in stochastic processes is supported by the central limit theorem, which at its turn is closely connected to the Cramér--Rao inequality as noted in \cite{Brown1982} and to several entropy-related inequalities \cite{Johnson2000, Johnson2004, Toscani2016(a)}, whose optimizing distribution is precisely the Gaussian density. Beyond the range of validity of the central limit theorem, the so-called stretched Gaussian densities have proven to be highly effective in describing a broad variety of physical systems (see for example \cite{Pluchino2008, Umarov2008, Amir2020, Tsallis2023}). Furthermore, the entire family of stretched Gaussian densities can be characterized as the family of minimizers of certain generalizations of the Cram\'er--Rao inequality. Consequently, the derivation of sharp inequalities involving information-theoretic functionals has become an active and important area of research, as seen in  \cite{Rioul2010, Ho2015, Sason2016, Yamano2021, Zozor2021, Wu2025}.

Since the most employed classical informational functionals are (Shannon and R\'enyi) entropies, absolute moments and Fisher information, the inequalities bounding these quantities have become essential in the theory. Some of the most well-known informational inequalities are the Stam inequality, introduced in \cite{Stam1959}, the moment-entropy inequality introduced in its first form in \cite{Shannon1948} and in modern form in the seminal work \cite{Wehrl1978} and the aforementioned Cram\'er-Rao inequality, which is actually an immediate consequence of the previous two inequalities but has been established before them in \cite{Rao1945, Cramer1946}. For a modern treatment of these inequalities and extensions of them, we refer the reader to \cite{Lutwak2004, Lutwak2005, Bercher2012} (see also references therein). More recently, these three inequalities have been extended by the authors and their collaborators to new functionals~\cite{Zozor2017,Puertas2025} and to a mirrored domain of validity for the entropic parameters~\cite{IP2025}, with the aid of Sundman-like transformations. Other extensions of these inequalities to more general classes of informational functionals depending either on the second derivative or on some incomplete weighted integrals of the density have been proved in~\cite{IP2025(b),IP2026b}. Actually, sharp upper bounds of the Stam-like product and moment-entropy product functionals, depending on regularity conditions satisfied by the probability density function, have been established throughout the use of these new functionals in \cite{IP2025(b)}. Moreover, sharp Stam-like and moment-entropy-like inequalities in the relative framework (that is, involing Kullback-Leibler and Rényi divergences) have also been established recently in \cite{IPT2025}.

The present short note is aimed at proving three optimal informational inequalities merging functionals depending on one and two probability density functions; that is, combining the non-relative and the relative framework. Their starting point is a new sharp inequality featuring the differential R\'enyi entropy, the R\'enyi divergence and the R\'enyi cross-entropy of a pair of probability density functions established by the authors in the recent work \cite{IP2026b}, generalizing the trivial equality between the Shannon cross-entropy and the sum of the Shannon entropy with the Kullback-Leibler divergence. More precisely, we prove:

$\bullet$ a Stam-like inequality involving the standard (one-parameter) R\'enyi entropy power, a scaling-invariant biparametric relative Fisher information and the usual R\'enyi cross-entropy.

$\bullet$ an inequality combining only functionals of moment type. We show that the product of the absolute moment (or its associated deviation) and the relative cumulative moment (taken at suitable powers) is optimally bounded from below by a functional called the cross-deviation of two probability density functions.

$\bullet$ an inequality combining only functionals of Fisher type. We prove that the product of the biparametric Fisher information of a density and the biparametric relative Fisher information of a pair or densities is optimally bounded from below by a new functional named generalized cross-Fisher information of the pair of densities.

We believe that one remarkable point of these inequalities is their sharpness. Indeed, all these three inequalities are sharp and their optimizers (that is, pairs of densities reaching the equality case) are given in terms of the stretched deformed Gaussian densities or other interesting special functions appearing in many applications, such as generalized trigonometric or hyperbolic functions and Beta distributions. Moreover, let us stress here that the optimizers of the first inequality of this paper correspond to the previously established minimizers for the generalized (triparametric) Stam inequality introduced in \cite{Puertas2025}, while the minimizers of the second inequality from the above list are actually a pair of first kind Beta distribution functions. Finally, the third inequality is also sharp, but its optimizing pairs rely on solutions to a complicated differential equation which, as far as we know, cannot be explicitly integrated.

A second interesting point of these inequalities is that they feature classical, well-studied informational functionals together with some new functionals whose mathematical expressions are derived in a similar way as the more usual R\'enyi cross-entropy is derived from the standard R\'enyi entropy. Note that the second and third inequalities only combine functionals of the same category, being the analogous one at the level of moments and of Fisher measures to the inequality involving three R\'enyi functionals introduced in \cite{IP2026}.

\section{A brief recall of definitions and results}

For the sake of completeness, we recall here some of the informational functionals and inequalities used in the sequel. In general, if nothing else is specified, we assume that $f$ and $h$ are probability density functions such that
\begin{equation}\label{cond:support}
	\begin{split}
		&{\rm supp}\,f={\rm supp}\,h=\overline{\Omega}, \quad \Omega=(c,d)\subseteq\Rset, \\
		&f(x)>0, \ h(x)>0, \quad {\rm for \ any} \ x\in\Omega,
	\end{split}
\end{equation}
where $\Omega$ can be either bounded or unbounded and $\overline{\Omega}$ denotes the closure of the set $\Omega$. All the integrals throughout the paper are assumed to be defined in the support of the integrand. We present next, in a sequence of dedicated paragraphs, the definitions of the informational functionals employed in the sequel.

\medskip

\noindent \textbf{R\'enyi entropy.} Given $\alpha\neq1$, the differential R\'enyi entropy of $\alpha$-order of a probability density function $\pdf$ is defined as
\begin{equation*}
	R_\alpha[\pdf] = \frac1{1-\alpha}\log\left( \int_\Rset\pdf^\alpha(x) \, dx \right).
\end{equation*}
In the limiting case $\alpha=1$ we recover the well-known differential Shannon entropy
\begin{equation*}
	\lim\limits_{\alpha\to 1}R_{\alpha}[\pdf]=S[\pdf]=-\int_\Rset \pdf(x) \, \log\pdf(x) \, dx.
\end{equation*}
A related functional is the \emph{R\'enyi entropy power}, defined by the exponential of the R\'enyi entropy, that is $N_{\alpha}[f]=\exp\{R_{\alpha}[f]\}$.

\medskip

\noindent \textbf{R\'enyi divergence.} Given $\alpha\neq1$, the differential R\'enyi divergence of $\alpha$-order of a pair of probability density functions $\pdf$ and $g$ satisfying \eqref{cond:support} is defined as
\begin{equation*}
	D_\alpha[\pdf||g] = \frac1{\alpha-1} \log\left( \int_\Rset\pdf^{\alpha}(x)g^{1-\alpha}(x) \, dx \right).
\end{equation*}
In the limiting case $\alpha=1$ we recover the well-known Kullback-Leibler divergence
\begin{equation*}
	\lim\limits_{\alpha\to 1}D_{\alpha}[\pdf||g]=D[\pdf||g]=\int_\Rset \pdf(x) \, \log\frac{\pdf(x)}{g(x)} \, dx.
\end{equation*}

\medskip

\noindent \textbf{R\'enyi cross-entropy.} Given $\alpha\neq1$, the differential R\'enyi cross-entropy of $\alpha$-order of a probability density function $\pdf$ relative to a probability density function $g$ is defined as
\begin{equation*}
	H_\alpha[\pdf;g] = \frac1{1-\alpha} \log\left( \int_\Rset \pdf(x)g^{\alpha-1}(x) \, dx \right).
\end{equation*}
In the limiting case $\alpha=1$ we recover the well-known differential Shannon cross-entropy
\begin{equation*}
	\lim\limits_{\alpha\to 1}H_{\alpha}[\pdf]=H[\pdf;g]=-\int_\Rset \pdf(x) \, \log g(x) \, dx.
\end{equation*}
Note that $H_{\alpha}[f;f]=R_{\alpha}[f]$.

\medskip

\noindent \textbf{The $(p,\lambda)$-Fisher information} is a generalization of the classical (non-parametric) Fisher information
\begin{equation*}
	F[f]=\int_\Rset f(x)\left(\frac{f'(x)}{f(x)}\right)^2\,dx
\end{equation*}
and it was introduced in \cite{Lutwak2005, Bercher2012, Bercher2012a}. Given $p>1$ and $\lambda\in\Rset\setminus\{0\}$, the $(p,\lambda)-$Fisher information of a probability density function $f$, which is assumed to be derivable on its support, is defined as
\begin{equation*}
F_{p,\lambda}[f]\: = \: \int_\Rset \left|f^{\lambda-2}(x) \, \frac{{\rm d}f}{{\rm d}x}(x)\right|^{p} \pdf(x) \,dx.
\end{equation*}
It is more convenient for informational inequalities to use the following form of the functional:
\begin{equation}\label{eq:FI}
\phi_{p,\lambda}[f]:=\left(F_{p,\lambda}[f]\right)^\frac{1}{p\lambda}.
\end{equation}

\medskip

\noindent \textbf{The relative $(p,\lambda)$-Fisher information.} This is an informational functional introduced in \cite{IPT2025} and obtained as the application of the $(p,\lambda)$-Fisher information to the relative differential escort transformation. Let $f$, $h$ be two probability density functions satisfying \eqref{cond:support}, being both derivable on their support, and let $(p,\lambda)\in\Rset^2$ be such that $p>1$ and $\lambda\neq0$. The \textit{relative $(p,\lambda)$-Fisher information} is defined as
\begin{equation*}
\phi_{p,\lambda}[f||h]:=F_{p,\lambda}[f||h]^\frac{1}{p\lambda},\quad {\rm with} \quad 	F_{p,\lambda}[f||h]:=\int_{\Rset}f^{1+p(\lambda-1)}(x)h^{-\lambda p}(x)\left|\frac {\rm {d}}{{\rm d}x}\left(\log\frac{f}{h}\right)(x)\right|^p\,dx.
\end{equation*}

\medskip

\noindent \textbf{The $p$-th absolute moment} of a probability density function $f$, with $p \geqslant 0$, is defined as
\begin{equation*}
\mu_p[f]= \int_\Rset |x|^p f(x) \,dx.
\end{equation*}
It is more common and convenient to consider the related quantity known as the $p$-th deviation, that is,
\begin{equation*}
\begin{split}
&\sigma_p[f] =\left(\int_\Rset |x|^p \, f(x) \,dx \right)^{\frac{1}{p}}, \quad {\rm for} \ p > 0, \\
&\sigma_0[f] = \lim\limits_{p \to 0} \sigma_p[f]  = \exp\left(\int_\Rset f(x) \, \log|x| \,dx \right).
\end{split}
\end{equation*}

\medskip

\noindent \textbf{The relative cumulative moment} is a functional extending to the relative framework the cumulative moments derived (see \cite{Puertas2025}) via the standard differential-escort transformations. Let $f$ and $h$ be two probability density functions satisfying \eqref{cond:support} and let $p>0$, $\alpha\in\Rset$. The \textit{relative cumulative moment} is defined as
\begin{equation*}
\mu_{p,\alpha}[f||h]=\int_{\Omega}\left|\int_{c}^xf(s)^{1-\alpha}h(s)^{\alpha}\,ds\right|^pf(x)\,dx.
\end{equation*}
We also introduce the associated $(p,\alpha)$-deviation by
\begin{equation}\label{eq:relcummom}
\sigma_{p,\alpha}[f||h]:=\mu_{p,\alpha}[f||h]^\frac1{p\alpha}.
\end{equation}

\medskip

\noindent \textbf{Some special functions.} We recall in this last paragraph of preliminary facts two special functions that are commonly related to minimizers of informational inequalities. The first class of them is the stretched Gaussians $g_{p,\lambda}$ given by
\begin{equation}\label{eq:gpl}
g_{p,\lambda}(x) \, = \, \frac{a_{p,\lambda}}{\exp_\lambda\left( |x|^{p^*} \right)}
\, = \, a_{p,\lambda}\, \exp_{2-\lambda}\left( - |x|^{p^*} \right),
\end{equation}
where $\exp_\lambda$ is the generalized Tsallis exponential
\begin{equation*}
\exp_\lambda(x) = \left( 1 + (1 - \lambda) \, x \right)_+^\frac1{1-\lambda}, \ \ \lambda \ne 1, \qquad \exp_1(x) \: \equiv \: \lim_{\lambda \to 1} \, \exp_\lambda(x) \: = \: \exp(x),
\end{equation*}
and the normalization constants $a_{p,\lambda}$ have a rather complicated explicit value which we omit here (see \cite{IP2025}). A second class of special functions involved in the minimizers of our inequalities are the $(p,q)$-generalized trigonometric and hyperbolic functions first defined, to the best of our knowledge, by \cite{Drabek1999}. The $(p,q)$-sine function $\sin_{p,q}$ is defined as the inverse of
\begin{equation*}
	\text{arcsin}_{p,q} (z)=\int_0^z(1-t^q)^{-\frac1p}dt, \quad p,q>1, \quad z\in[0,1).
\end{equation*}
The cosine function is defined as the derivative of the sine function
\begin{equation*}
	\cos_{p,q}(z)=\frac{d}{dz}\sin_{p,q}(z).
\end{equation*}
From these definitions, and applying the inverse function rule, the following identity is obtained:
\begin{equation}\label{eq:pyth}
	\cos_{p,q}^p(z)+\sin_{p,q}^q(z)=1.
\end{equation}
The hyperbolic sine function $\sinh_{p,q}$ is defined in a similar way as the inverse of
\begin{equation*}
	\text{arcsinh}_{p,q} (z)=\int_0^z(1+t^q)^{-\frac1p}dt.
\end{equation*}
The hyperbolic cosine function is defined as the derivative of the hyperbolic sine function
	\begin{equation*}
		\cosh_{p,q}(z)=\frac{d}{dz}\sinh_{p,q}(z) \, ,
	\end{equation*}
and it follows that
	\begin{equation}\label{eq:pythh}
		\cosh^p_{p,q}(z)-\sinh^q_{p,q}(z)=1.
	\end{equation}

\medskip

\noindent \textbf{A sharp inequality between R\'enyi functionals.} The following sharp inequality has been established by the authors in their recent paper \cite{IP2026}. We recall it here, as it is the starting point of the main results of this paper. Let $\alpha,\beta,\gamma$ be three real numbers satisfying the following equality:
\begin{equation}\label{eq:relation}
(\alpha-\beta)(\alpha-\gamma)=(\alpha-1)^2.
\end{equation}
If $\alpha>\beta$ then the following inequality
\begin{equation}\label{eq:ineqRRR}
R_\alpha[f]+D_\beta[f||h]\leqslant H_{\gamma}[f;h]
\end{equation}
holds true for any probability density functions satisfying \eqref{cond:support}. If $\alpha<\beta$ the inequality is reversed. The equality is attained in \eqref{eq:ineqRRR} if and only if
\begin{equation}\label{eq:main}
	h(x)\propto f^{\frac{\beta-1}{\beta-\alpha}}(x),
\end{equation}
where the notation $h(x)\propto g(x)$ indicates that $h(x)=cg(x)$ for some constant $c\neq0$.

\medskip

These are the most common informational functionals and special functions that will appear in the forthcoming inequalities. The main results of this paper can be seen as further applications of the inequality \eqref{eq:ineqRRR} combined with other previously established results. Other more specialized preliminary results and functionals will be recalled at the point where they are employed or referred. For simplicity, we also introduce the following notation
\begin{equation}\label{la}
\lambda_{\alpha}:=1+\alpha(\lambda-1),
\end{equation}
where $\lambda$ and $\alpha$ are real numbers.

\section{A combined Stam-like inequality}

The first result of this work is a Stam-like inequality which mixes a one-parameter R\'enyi entropy power, a biparametric Fisher measure in the relative framework and a R\'enyi cross entropy.
\begin{theorem}\label{th:Stam-like}
Let $p\geqslant1$, $\lambda$, $\alpha$ and $\gamma$ be real numbers satisfying the following conditions:
\begin{equation}\label{eq:conds}
\lambda>\frac{1}{1+p^*}, \quad \alpha(2-\lambda)<1, \quad (\alpha(2-\lambda)-1))(\alpha-\gamma)=(\alpha-1)^2.
\end{equation}
Then there exists a positive constant $K$ depending on the previous parameters such that the following inequality
\begin{equation}\label{ineq:Stam-like}
N_{\alpha}[f]\phi_{p,\lambda\alpha}[f||h]\geqslant K\exp\{H_{\gamma}[f;h]\}
\end{equation}
holds true for any derivable probability density functions $f$ and $h$ satisfying \eqref{cond:support} and such that $(f/h)^{\alpha}$ is an absolutely continuous function on $\overline{\Omega}$. The inequality is sharp and the equality is achieved whenever
\begin{equation}\label{eq:opt}
N_{\frac{1}{\alpha^*}}[g_{p,\lambda}]<\infty, \quad \alpha^*:=\frac{\alpha}{\alpha-1}, \quad {\rm if} \ \alpha\in\real\setminus\{0,1\},
\end{equation}
or
\begin{equation}\label{eq:opt1}
\lambda>1, \quad {\rm if} \ \alpha=1.
\end{equation}
Moreover, the functions achieving equality can be expressed in terms of $g_{p,\lambda}$ and generalized trigonometric and hyperbolic functions.
\end{theorem}
Let us mention here that the optimal constant $K$ (in the cases when the inequality is sharp) can be explicitly computed in terms of the optimal constant of the biparametric Stam inequality (see for example \cite{Lutwak2005, Bercher2012}).
\begin{proof}
On the one hand, by taking exponentials in \eqref{eq:ineqRRR} and particularizing $\beta=\lambda_{\alpha}$, we find
\begin{equation}\label{eq:interm1}
N_{\alpha}[f]e^{D_{\lambda_{\alpha}}[f||h]}\geqslant e^{H_{\gamma}[f;h]},
\end{equation}
where the assumption \eqref{eq:relation} is equivalent to the third condition in \eqref{eq:conds}. Let us notice that the condition $\alpha<\beta$ required for the inequality sign in \eqref{eq:interm1} is equivalent to
$$
\alpha<\lambda_{\alpha}=1+\alpha(\lambda-1),
$$
which leads to the second condition in \eqref{eq:conds}. On the other hand, let us recall the Stam-like inequality in the relative framework derived in \cite[Theorem 4.1]{IPT2025}, which in a particular case states that for any $p\geqslant1$, $\lambda>\frac{1}{1+p^*}$ and $(f/h)^{\alpha}$ absolutely continuous, then
\begin{equation}\label{eq:interm2}
e^{-D_{\lambda_{\alpha}}[f||h]}\phi_{p,\lambda\alpha}[f||h]\geq\alpha^{-\frac{1}{\alpha\lambda}}\left(K_{p,\lambda}^{(1)}\right)^\frac1{\alpha},
\end{equation}
where the constant $K_{p,\lambda}^{(1)}$ is the optimal bound of the generalized Stam inequality (see~\cite{Lutwak2005}) adopting the notation defined in~\cite{Puertas2025, IPT2025}. By multiplying \eqref{eq:interm1} and \eqref{eq:interm2}, we arrive at \eqref{ineq:Stam-like}, as claimed.

\medskip

\noindent \textbf{Sharpness and optimizers.} In \eqref{eq:interm1}, the equality is achieved, according to \eqref{eq:main}, for pairs $(f,h)$ of densities such that
\begin{equation}\label{eq:interm3}
f(x)\propto h(x)^{\frac{\lambda_{\alpha}-\alpha}{\lambda_{\alpha}-1}}=h(x)^{\frac{1+\alpha(\lambda-2)}{\alpha(\lambda-1)}}.
\end{equation}
The inequality \eqref{eq:interm2} is sharp whenever the conditions \eqref{eq:opt} or \eqref{eq:opt1} are satisfied, as proved in \cite[Theorem 4.1]{IPT2025}. Moreover, letting either
\begin{equation}\label{eq:defr}
r:=N_{\frac{1}{\alpha^*}}[g_{p,\lambda}]<\infty,  \quad {\rm if} \ \alpha\neq1
\end{equation}
or $r$ to be the (finite) length of the compactly supported function $g_{p,\lambda}$ if $\alpha=1$ and $\lambda>1$, the inequality \eqref{eq:interm2} is minimized by the following function expressed in terms of the stretched Gaussians $g_{p,\lambda}$ defined in \eqref{eq:gpl}:
\begin{equation}\label{eq:interm4}
f_{\rm min}(x)=\left[rg_{p,\lambda}(ry(x))\right]^{\frac{1}{\alpha}}h(x), \quad y'(x)=\left[rg_{p,\lambda}(ry(x))\right]^{-1/\alpha^*}h(x).
\end{equation}
By combining \eqref{eq:interm3} and \eqref{eq:interm4}, we obtain that
$$
h(x)^{\frac{1+\alpha(\lambda-2)}{\alpha(\lambda-1)}}\propto\left[rg_{p,\lambda}(ry(x))\right]^{\frac{1}{\alpha}}h(x),
$$
or, equivalently,
\begin{equation}\label{eq:interm7}
rg_{p,\lambda}(ry(x))\propto h(x)^{\frac{1-\alpha}{\lambda-1}}.
\end{equation}
We next replace the previous equivalence (modulo a normalization constant) in the expression of $y'(x)$ in \eqref{eq:interm4} to find
\begin{equation}\label{eq:interm5}
y'(x)\propto\left[rg_{p,\lambda}(ry(x))\right]^{\kappa}, \quad \kappa:=-\frac{1}{\alpha^*}+\frac{\lambda-1}{1-\alpha}.
\end{equation}
The differential equation \eqref{eq:interm5} can be integrated in terms of special functions. Assume first that $\lambda<1$. Then we have (ignoring the proportionality constant in order to replace $\propto$ by equality, for simplicity)
\begin{equation*}
\begin{split}
x&=\int_0^{y(x)}\frac{d \overline y}{r^{\kappa}g_{p,\lambda}^{\kappa}(r\overline  y)}=\frac{1}{(ra_{p,\lambda})^{\kappa}}\int_0^{y(x)}\frac{d\overline y}{\left(1+(1-\lambda)(r\overline y)^{p^*}\right)^{\frac{\kappa}{\lambda-1}}}\\
&=\frac{1}{(ra_{p,\lambda})^{\kappa}}\int_0^{|1-\lambda|^{\frac{1}{p^*}}ry(x)}\frac{|1-\lambda|^{-\frac{1}{p^*}}}{r(1+s^{p^*})^{\frac{\kappa}{\lambda-1}}}\,ds\\
&=\frac{1}{r^{\kappa+1}a_{p,\lambda}^{\kappa}|1-\lambda|^{\frac{1}{p^*}}}{\rm arcsinh}_{\frac{\lambda-1}{\kappa},p^*}\left(|1-\lambda|^{\frac{1}{p^*}}ry(x)\right).
\end{split}
\end{equation*}
A completely similar calculation can be performed in the case $\lambda>1$ by simply replacing the generalized hyperbolic function by the generalized trigonometric function to find
$$
x=\frac{1}{r^{\kappa+1}a_{p,\lambda}^{\kappa}|1-\lambda|^{\frac{1}{p^*}}}\arcsin_{\frac{\lambda-1}{\kappa},p^*}\left(|1-\lambda|^{\frac{1}{p^*}}ry(x)\right).
$$
We infer by inverting the previous inequalities that
\begin{equation}\label{eq:interm6}
|1-\lambda|^{\frac{1}{p^*}}ry(x)=\begin{cases}
                                   \sinh_{\frac{\lambda-1}{\kappa},p^*}\left(r^{\kappa+1}a_{p,\lambda}^{\kappa}|1-\lambda|^{\frac{1}{p^*}}x\right), & \mbox{if } \lambda<1,\\
                                   \sin_{\frac{\lambda-1}{\kappa},p^*}\left(r^{\kappa+1}a_{p,\lambda}^{\kappa}|1-\lambda|^{\frac{1}{p^*}}x\right), & \mbox{if } \lambda>1.
                                 \end{cases}
\end{equation}
We then derive from \eqref{eq:interm6} and \eqref{eq:interm7} that, once more choosing $\lambda<1$ as sample case,
\begin{equation*}
\begin{split}
h(x)^{\frac{1-\alpha}{\lambda-1}}&\propto ra_{p,\lambda}\left(1+(1-\lambda)r^{p^*}y(x)^{p^*}\right)^{\frac{1}{\lambda-1}}\\
&=ra_{p,\lambda}\left(1+\sinh_{\frac{\lambda-1}{\kappa},p^*}^{p^*}\left(r^{\kappa+1}a_{p,\lambda}^{\kappa}(1-\lambda)^{\frac{1}{p^*}}x\right)\right)^{\frac{1}{\lambda-1}}\\
&=ra_{p,\lambda}\cosh_{\frac{\lambda-1}{\kappa},p^*}^{\frac{1}{\kappa}}\left(r^{\kappa+1}a_{p,\lambda}^{\kappa}(1-\lambda)^{\frac{1}{p^*}}x\right).
\end{split}
\end{equation*}
We thus finally obtain that
\begin{equation}\label{eq:interm8}
h(x)=\begin{cases}
       (ra_{p,\lambda})^{\frac{\lambda-1}{1-\alpha}}\cosh_{\frac{\lambda-1}{\kappa},p^*}^{\frac{\lambda-1}{\kappa(1-\alpha)}}(\overline{x}), & \mbox{if } \lambda<1, \\[2mm]
       (ra_{p,\lambda})^{\frac{\lambda-1}{1-\alpha}}\cos_{\frac{\lambda-1}{\kappa},p^*}^{\frac{\lambda-1}{\kappa(1-\alpha)}}(\overline{x}), & \mbox{if } \lambda>1,
     \end{cases}
\end{equation}
with
$$
\overline{x}:=r^{\kappa+1}a_{p,\lambda}^{\kappa}(1-\lambda)^{\frac{1}{p^*}}x.
$$
The corresponding optimizer $f$ is deduced by simply gathering \eqref{eq:interm3} and \eqref{eq:interm8}. This completes the proof of the optimality of the inequality \eqref{ineq:Stam-like}.
\end{proof}

\noindent \textbf{Remarks. 1. A simpler particular case.} There is a particular case when $\kappa=0$, that is,
$$
0=\frac{\lambda-1}{1-\alpha}-\frac{1}{\alpha^*}=\frac{\lambda-1}{1-\alpha}-\frac{\alpha-1}{\alpha},
$$
which implies
\begin{equation}\label{eq:part}
\lambda=1-\frac{(\alpha-1)^2}{\alpha}.
\end{equation}
In this case there is no change of independent variable in \eqref{eq:interm4} and the pair of minimizers is given by powers of $g_{p,\lambda}$, following readily from \eqref{eq:interm7} (for $h$) and then \eqref{eq:interm3} (for $f$). More precisely,
$$
h(x)\propto g_{p,\lambda}^{\frac{\alpha-1}{\alpha}}(rx)=g_{p,\lambda}^{\frac{1}{\alpha^*}}(rx), \quad f(x)\propto g_{p,\lambda}(rx),
$$
hence in the particular case \eqref{eq:part} we can say, modulo a scaling, that the equality is achieved for $g_{p,\lambda}$ and a power of $g_{p,\lambda}$. Note that, up to a scaling change, $g_{p,\lambda}^\frac{1}{\alpha*}=g_{p,\overline\lambda}$, with $\overline\lambda=\frac{\lambda\alpha-1}{\alpha-1}$.

\medskip

\noindent \textbf{2.} It is worth mentioning that the family of pairs of optimizers of the inequality~\eqref{ineq:Stam-like} are analogous to the single minimizers of the triparametric inequalities of Stam, moment-entropy and Cr\'amer-Rao types. These minimizers involve the generalized trigonometric densities as well as the stretched Gaussian densities $g_{p,\lambda}$, see~\cite{Puertas2025}.

\section{An inequality between moment-like functionals}

In this section, we establish an inequality involving three informational functionals of moment type. Taking into account the structure with levels introduced in \cite{IP2026b}, this inequality is a kind of parallel one to \eqref{eq:ineqRRR} but at the level of moments (or, equivalently, deviations). In order to state the following result, let us recall first the following functional called the \emph{cross-deviation} of $f$ relative to $h$. Let $p,\gamma\in\Rset$ and let $f$, $h$ be two probability density functions. The \emph{cross-deviation} of $f$ relative to $h$ is defined as
\begin{equation}\label{eq:cross-moment}
	\sigma_{p,\gamma}[f;h]:=\left(\int_\mathbb{R}f^{2-\gamma}(x)h^{\gamma-1}(x)\,|x|^p\,dx\right)^\frac1{p}.
\end{equation}
This cross-deviation has been introduced in \cite{IP2026} as a transported functional through the up transformation starting from the R\'enyi cross-entropy.
\begin{theorem}\label{th:SSS}
Let $(p,\lambda,\xi,a,\alpha,\gamma)\in\Rset^6$ be real numbers such that
\begin{equation}\label{eq:condSSS1}
p^*\geqslant0, \quad \lambda>\frac{1}{1+p^*}, \quad \xi>0, \quad a\in\Rset\setminus\{2\}
\end{equation}
and the following conditions are fulfilled:
\begin{equation}\label{eq:condSSS2}
\alpha-1>\xi(\lambda-1), \quad [\alpha-1-\xi(\lambda-1)](\alpha-\gamma)=(\alpha-1)^2.
\end{equation}
Then, there is $K>0$ such that the following inequality
\begin{equation}\label{ineq:SSS}
\sigma_{\frac{\alpha-1}{2-a}}^{\frac{1}{2-a}}[f]\sigma_{p^*,\xi}[f||h]\geq K\sigma_{\frac{\gamma-1}{2-a},\gamma}^{\frac{1}{2-a}}[f;h].
\end{equation}
Under the conditions \eqref{eq:opt} or \eqref{eq:opt1}, the inequality \eqref{ineq:SSS} is sharp and the equality case is achieved on an explicit pair of minimizers $(f,h)$.
\end{theorem}
Let us stress here that the optimal constant in \eqref{ineq:SSS} can be explicitly expressed in terms of the optimal constant of the moment-entropy inequality (see for example \cite{Lutwak2004, Lutwak2005, Bercher2012a}), as it can be readily seen from the proof.
\begin{proof}
On the one hand, for any $\xi>0$ and $p$, $\lambda\in\Rset$ satisfying the conditions related to them in \eqref{eq:condSSS1}, the following moment-entropy inequality in the relative framework
\begin{equation}\label{ineq:EMrel}
e^{D_{\lambda_{\xi}}}[f||h]\sigma_{p^*,\xi}[f||h]\geq\left(K_{p,\lambda}^{(0)}\right)^{\frac{1}{\xi}},
\end{equation}
has been established for any pair of densities $(f,h)$ satisfying \eqref{cond:support} in \cite[Theorem 4.1]{IPT2025}. We recall that $\sigma_{p^*,\xi}[f||h]$ designs the relative $(p^*,\xi)$-deviation defined in \eqref{eq:relcummom}, $\lambda_{\xi}$ is defined in \eqref{la} and the constant $K_{p,\lambda}^{(0)}$ is the optimal constant in the moment-entropy inequality according to the notation in \cite{IPT2025}. On the other hand, given the real numbers $\alpha$, $\beta$ and $\gamma$ satisfying $\alpha>\beta$ and \eqref{eq:relation}, the inequality \begin{equation}\label{ineq:up}
\left(\sigma_\frac{\alpha-1}{2-a}[f]\right)^\frac1{a-2}e^{D_\beta[f||h]}\leqslant\left(\sigma_{\frac{\gamma-1}{2-a},\gamma}[f;h]\right)^{\frac{1}{a-2}}.
	\end{equation}
holds true for any $a\in\Rset\setminus\{2\}$ and any pair of densities $(f,h)$ satisfying \eqref{cond:support}, according to \cite[Theorem 3.5]{IP2026}. We particularize $\beta=\lambda_{\xi}$ in \eqref{ineq:up} and observe that \eqref{eq:relation} and $\alpha>\beta=\lambda_{\xi}$ reduce exactly to \eqref{eq:condSSS2} after immediate algebraic manipulations. We can thus write \eqref{ineq:up} equivalently as
\begin{equation}\label{eq:interm9}
e^{-D_{\lambda_{\xi}}[f||h]}\geqslant\left(\sigma_\frac{\alpha-1}{2-a}[f]\right)^\frac1{a-2}\left(\sigma_{\frac{\gamma-1}{2-a},\gamma}[f;h]\right)^{\frac{1}{2-a}}.
\end{equation}
By multiplying \eqref{ineq:EMrel} and \eqref{eq:interm9}, we readily obtain the inequality \eqref{ineq:SSS}, with the optimal constant
$$
K:=\left(K_{p,\lambda}^{(0)}\right)^{\frac{1}{\xi}},
$$
thus depending directly on the optimal constant of the entropy-moment inequality established in \cite{Lutwak2004, Lutwak2005}.

\medskip

\noindent \textbf{Sharpness and optimizers.} The equality is achieved in the inequalities \eqref{ineq:up} and its equivalent form \eqref{eq:interm9} for pairs of densities $(f,h)$ satisfying (see \cite[Theorem 3.5]{IP2026})
\begin{equation}\label{eq:interm10}
h(x)\propto|x|^{\frac{1-\alpha}{(2-a)(\alpha-\beta)}}f(x)=|x|^{\frac{1-\alpha}{(2-a)(\alpha-1-\xi(\lambda-1))}}f(x).
\end{equation}
Under the conditions \eqref{eq:opt}, respectively \eqref{eq:opt1}, the minimizers of the inequality \eqref{ineq:EMrel} are given by (see \cite[Theorem 4.1]{IPT2025})
\begin{equation}\label{eq:interm11}
f(x)=[rg_{p,\lambda}(ry(x))]^{\frac{1}{\alpha}}h(x), \quad y'(x)=[rg_{p,\lambda}(ry(x))]^{-\frac{1}{\alpha^*}}h(x),
\end{equation}
recalling that $\alpha^*=\alpha/(\alpha-1)$ is the H\"older conjugate exponent to $\alpha$ (with the convention $\alpha^*=\infty$ if $\alpha=1$) and that the scaling factor $r$ is defined in \eqref{eq:defr} if $\alpha\neq1$ and as the finite length of the support of $g_{p,\lambda}$ given by \eqref{eq:opt1} if $\alpha=1$. Combining \eqref{eq:interm10} and \eqref{eq:interm11}, we obtain that
$$
f(x)\propto[rg_{p,\lambda}(ry(x))]^{\frac{1}{\alpha}}|x|^{\frac{1-\alpha}{(2-a)(\alpha-1-\xi(\lambda-1))}}f(x),
$$
which, after obvious simplifications and recalling the definition of the functions $g_{p,\lambda}$, leads to
\begin{equation}\label{eq:interm13}
ra_{p,\lambda}\left(1+(1-\lambda)r^{p^*}|y(x)|^{p^*}\right)^{\frac{1}{\lambda-1}}\propto|x|^{\theta}, \quad \theta:=\frac{\alpha(1-\alpha)}{(a-2)(\alpha-1-\xi(\lambda-1))}.
\end{equation}
Further algebraic manipulations then give
$$
|ry(x)|^{p^*}\propto\frac{1}{1-\lambda}\left[\left(\frac{|x|^{\theta}}{ra_{p,\lambda}}\right)^{\lambda-1}-1\right],
$$
which is equivalent to
\begin{equation}\label{eq:interm12}
y(x)\propto\frac{1}{r}\left|\left(\frac{|x|^{\theta}}{ra_{p,\lambda}}\right)^{\lambda-1}-1\right|^{\frac{1}{p^*}}.
\end{equation}
We now infer from the change of variable in \eqref{eq:interm11} that
$$
h(x)=[rg_{p,\lambda}(ry(x))]^{\frac{1}{\alpha^*}}y'(x)
$$
which, taking into account \eqref{eq:interm13} and \eqref{eq:interm12}, gives
\begin{equation}\label{eq:minimh}
h(x)=|x|^{\frac{\theta}{\alpha^*}}y'(x)\propto
\left|\left(\frac{|x|^{\theta}}{ra_{p,\lambda}}\right)^{\lambda-1}-1\right|^{-\frac{1}{p}}|x|^{\theta\left(\lambda-1+\frac{1}{\alpha^*}\right)-2}\,x.
\end{equation}
Finally, we can also deduce the optimizer $f(x)$ by gathering \eqref{eq:interm11} and \eqref{eq:minimh}, more precisely
\begin{equation}\label{eq:minimf}
f(x)\propto
\left|\left(\frac{|x|^{\theta}}{ra_{p,\lambda}}\right)^{\lambda-1}-1\right|^{-\frac{1}{p}}|x|^{\theta\lambda-2}\,x,
\end{equation}
mentioning here that, in \eqref{eq:minimh} and \eqref{eq:minimf}, we have possibly neglected a minus sign by introducing it in the constant of proportionality.
\end{proof}

\noindent \textbf{Remarks. 1.} Note that the optimizers of the inequality~\eqref{ineq:SSS} correspond to a pair of first kind generalized Beta probability density functions (see~\cite{McDonald_Xu1995}).

\medskip

\noindent \textbf{2. The critical exponent $a=2$.} An inequality in the style of \eqref{ineq:SSS} can be deduced as well for the critical exponent $a=2$. In this case, \eqref{ineq:up} is replaced by the inequality
\begin{equation}\label{ineq:up2}
\left\langle e^{(1-\alpha)x}\right\rangle_f^{\frac{1}{1-\alpha}}e^{D_{\beta}[f||h]}\leqslant\sigma_{\gamma}^{(E)}[f;h],
\end{equation}
established in \cite[Theorem 3.5]{IP2026}, where
\begin{equation*}\label{eq:cross-moment2}
\sigma^{(E)}_{\gamma}[f;g]:=\left(\int_{\Rset}f^{2-\gamma}(x)g^{\gamma-1}(x)e^{(1-\gamma)x}\,dx\right)^{\frac{1}{1-\gamma}}.
\end{equation*}
is a kind of exponential cross-moment. Once more, particularizing $\beta=\lambda_{\xi}$ and eliminating the Kullback-Leibler divergence between \eqref{ineq:EMrel} and \eqref{ineq:up2} leads to an inequality analogous to \eqref{ineq:SSS}, but involving deviations based on exponential weights in the integrals (such as \eqref{eq:cross-moment2}) instead of powers of the form $|x|^{\theta}$. We leave the details of this deduction to the reader, as both the proof and the study of the optimizers is a straightforward adaptation of the proof of Theorem \ref{th:SSS}.

\section{An inequality between Fisher measures}

We complete the panorama of the inequalities established in this work with an inequality involving only Fisher-like measures, which is an analogous result at the level of Fisher measures to the outcome of \eqref{eq:ineqRRR} (at the level of R\'enyi entropies and divergences) and \eqref{ineq:SSS} (at the level of moments or deviations). We prevent the reader that this inequality is technically more involved, especially when discussing its sharpness and optimizers. Before stating the inequality, let us recall the following definition of the \emph{generalized cross-Fisher information}. If $a,b,c$ are real numbers and $f$ is a derivable probability density function, the $(a,b,c)$-cross-Fisher information of $f$ relative to a probability density $h$ is defined as
\begin{equation}\label{eq:CF}
	\phi^{\rm (cr)}_{a,b,c}[f;h]:=\left(\int_{\mathbb R}f^{1+(a-2)c}(x)\left(\frac{f(x)}{h(x)}\right)^{c}|f'(x)|^{bc}\,dx\right)^{\frac1{c}}.
\end{equation}
With this definition, we have all the needed notions in order to state the last result of this paper.
\begin{theorem}\label{th:FFF}
Let $(\alpha,\beta,a,b,\lambda,p,\gamma)\in\Rset^7$ be such that $p\geq1$, $\alpha>0$, $a\neq0$, $a\neq2b$ and the conditions \eqref{eq:conds} are satisfied. Then, there exists $K>0$ such that the following inequality
\begin{equation}\label{ineq:FFF}
\phi_{(1-\alpha)b,2-\frac{a}{b}}^{2b-a}[f]\phi_{p,\lambda\alpha}[f||h]\geq K\phi_{2-a,b,1-\gamma}^{(cr)}[f;h],
\end{equation}
holds true for any pair of derivable probability density functions $(f,h)$ satisfying \eqref{cond:support} and such that $(f/h)^{\alpha}$ is an absolutely continuous function in $\Rset$. The inequality \eqref{ineq:FFF} is sharp provided that conditions \eqref{eq:opt} or \eqref{eq:opt1} are satisfied, but the minimizers are not explicit and depend on solutions of a differential equation.
\end{theorem}
The optimal constant $K$ in the inequality \eqref{ineq:FFF} (in the case when it is sharp) can be expressed in terms of the optimal constant of the biparametric Stam inequality, and is in fact the same as for the inequality \eqref{ineq:Stam-like}.
\begin{proof}
On the one hand, since we are under exactly the same assumptions, we can particularize $\beta=\lambda_{\alpha}$ in \cite[Theorem 4.1]{IPT2025} and establish the inequality \eqref{eq:interm2}. On the other hand, since $a\neq0$, $a\neq2b$ and the condition $\alpha<\beta=\lambda_{\alpha}$ is satisfied, as it is equivalent to the second condition in \eqref{eq:conds}, we infer from \cite[Theorem 3.3]{IP2026} that
\begin{equation}\label{ineq:bip-down}
\phi_{(1-\alpha)b,2-\frac ab}^{2b-a}[f]e^{D_{\lambda_{\alpha}}[f||h]}\geqslant\phi^{\rm (cr)}_{2-a,b,1-\gamma}[f;h],
\end{equation}
under the assumptions \eqref{eq:conds}. We can thus multiply the inequalities \eqref{eq:interm2} and \eqref{ineq:bip-down} in order to derive \eqref{ineq:FFF}, and the optimal constant (whenever the conditions \eqref{eq:opt} or \eqref{eq:opt1} are fulfilled) is exactly the same as for the inequality \eqref{ineq:Stam-like}.

\medskip

\noindent \textbf{Sharpness and optimizers.} On the one hand, the equality is achieved in the inequality \eqref{ineq:bip-down} for pairs of densities $(f,h)$ such that
\begin{equation}\label{eq:interm14}
h(x)\propto f(x)^{A}|f'(x)|^{B}, \quad A=1+\frac{a(1-\alpha)}{\alpha(2-\lambda)-1}, \quad B=\frac{b(\alpha-1)}{\alpha(2-\lambda)-1},
\end{equation}
where $A$ and $B$ are well defined since the condition \eqref{eq:conds} implies that $\alpha(2-\lambda)-1=\alpha-\beta\neq0$. On the other hand, the equality in \eqref{eq:interm2} is achieved for densities $f(x)$ given by \eqref{eq:interm4}. Thus, by replacing $h$ with the proportional right-hand side in \eqref{eq:interm14}, we deduce after obvious simplifications that
\begin{equation}\label{eq:interm15}
\left[rg_{p,\lambda}(ry(x))\right]^{\frac{1}{\alpha}}\propto f(x)^{1-A}|f'(x)|^{-B},
\end{equation}
where $r$ is defined in \eqref{eq:defr}, whenever \eqref{eq:opt} or \eqref{eq:opt1} is in force. Moreover, replacing $h$ from \eqref{eq:interm14} in the expression of the change of variable in \eqref{eq:interm2}, we find
\begin{equation}\label{eq:interm16}
\left[rg_{p,\lambda}(ry(x))\right]^{\frac{1}{\alpha^*}}y'(x)\propto f(x)^{A}|f'(x)|^{B}.
\end{equation}
By multiplying \eqref{eq:interm15} and \eqref{eq:interm16}, we find that
\begin{equation}\label{eq:interm18}
f(x)\propto rg_{p,\lambda}(ry(x))y'(x).
\end{equation}
By taking derivatives, we further get that
$$
f'(x)\propto r^2g_{p,\lambda}'(ry(x))(y'(x))^2+rg_{p,\lambda}(ry(x))y''(x).
$$
We next insert these equivalences of $f(x)$ and $f'(x)$ into \eqref{eq:interm15} in order to derive a closed differential equation of second order for $y(x)$. More precisely, \eqref{eq:interm15} is equivalently written as
\begin{equation}\label{eq:interm17}
f(x)^{\frac{A-1}{B}}|f'(x)|\propto\left[rg_{p,\lambda}(ry(x))\right]^{-\frac{1}{B\alpha}},
\end{equation}
and, by taking into account the proportionality expressions for $f(x)$ and $f'(x)$, \eqref{eq:interm17} becomes
$$
\left[rg_{p,\lambda}(ry(x))\right]^{-\frac{1}{B\alpha}+\frac{1-A}{B}}\propto(y'(x))^{\frac{A-1}{B}}\left[r^2g_{p,\lambda}'(ry(x))y'(x)^2+
rg_{p,\lambda}(ry(x))y''(x)\right].
$$
Noticing that
$$
\frac{1-A}{B}=\frac{a}{b},
$$
we have thus obtained a second order differential equation, that can be written in a simpler form as follows:
\begin{equation}\label{eq:diff}
\left[rg_{p,\lambda}(ry(x))\right]^{-\frac{1}{B\alpha}+\frac{a}{b}-1}\propto(y'(x))^{\frac{2b-a}{b}}
\left[\frac{rg_{p,\lambda}'(ry(x))}{g_{p,\lambda}(ry(x))}+\frac{y''(x)}{y'(x)^2}\right].
\end{equation}
It is quite obvious that \eqref{eq:diff} cannot be explicitly integrated. However, we can conclude by stressing that any solution $y(x)$ of the equation \eqref{eq:diff} produces a pair of optimizers $(f,h)$ defined by \eqref{eq:interm18} and \eqref{eq:interm14} respectively.
\end{proof}

\subsection*{Acknowledgements}

R. G. I. is partially supported by the project PID2024-160967NB-I00 (AEI) funded by the Ministry of Science, Innovation and Universities of Spain and FEDER/EU. D. P.-C. is partially supported by the project PID2023-153035NB-100 (AEI) funded by the Ministry of Science, Innovation and Universities of Spain and “ERDF/EU A way of making Europe” and by Project 2025/SOLCON-160677 funded by Universidad Rey Juan Carlos.

\bigskip

\noindent \textbf{Data availability} Our manuscript has no associated data.

\bigskip

\noindent \textbf{Competing interest} The authors declare that there is no competing interest.

		\bibliographystyle{elsarticle-num}
		\bibliography{refsbis}

\begin{thebibliography}{10}
\expandafter\ifx\csname url\endcsname\relax
  \def\url#1{\texttt{#1}}\fi
\expandafter\ifx\csname urlprefix\endcsname\relax\def\urlprefix{URL }\fi
\expandafter\ifx\csname href\endcsname\relax
  \def\href#1#2{#2} \def\path#1{#1}\fi

\bibitem{Brown1982}
L.~D. Brown, A proof of the central limit theorem motivated by the
  {C}ram{\'e}r-{R}ao inequality, in: G.~Kallianpur, P.~R. Krishnaiah, J.~K.
  Ghosh (Eds.), Statistics and Probability: Essays in Honor of C. R. Rao,
  North-Holland, Amsterdam, 1982, pp. 141--148.

\bibitem{Johnson2000}
O.~Johnson, Entropy inequalities and the central limit theorem, Stochastic
  Processes and Their Applications 88~(2) (2000) 291--304.

\bibitem{Johnson2004}
O.~Johnson, A.~Barron, Fisher information inequalities and the central limit
  theorem, Probability Theory and Related Fields 129~(3) (2004) 391--409.

\bibitem{Toscani2016(a)}
G.~Toscani, Entropy inequalities for stable densities and strengthened central
  limit theorems, Journal of Statistical Physics 165~(2) (2016) 371--389.

\bibitem{Pluchino2008}
A.~Pluchino, A.~Rapisarda, C.~Tsallis, A closer look at the indications of
  $q$-generalized {C}entral {L}imit {T}heorem behavior in quasi-stationary
  states of the {HMF} model, Physica A: Statistical Mechanics and its
  Applications 387~(13) (2008) 3121--3128.

\bibitem{Umarov2008}
S.~Umarov, C.~Tsallis, S.~Steinberg, On a $q$-central limit theorem consistent
  with nonextensive statistical mechanics, Milan Journal of Mathematics 76~(1)
  (2008) 307--328.

\bibitem{Amir2020}
A.~Amir, An elementary renormalization-group approach to the generalized
  central limit theorem and extreme value distributions, Journal of Statistical
  Mechanics: Theory and Experiment 2020~(1) (2020) 013214.

\bibitem{Tsallis2023}
C.~Tsallis, Introduction to Nonextensive Statistical Mechanics: Approaching a
  Complex World, 2nd Edition, Springer, 2023.
\newblock \href {https://doi.org/10.1007/978-3-030-79569-6}
  {\path{doi:10.1007/978-3-030-79569-6}}.

\bibitem{Rioul2010}
O.~Rioul, Information theoretic proofs of entropy power inequalities, IEEE
  Transactions on Information Theory 57~(1) (2010) 33--55.

\bibitem{Ho2015}
S.-W-Ho, S.~Verd{\'u}, Convexity/concavity of {R}{\'e}nyi entropy and
  $\alpha$-mutual information, in: 2015 IEEE International Symposium on
  Information Theory (ISIT), IEEE, 2015, pp. 745--749.

\bibitem{Sason2016}
I.~Sason, S.~Verd{\'u}, $f$-divergence inequalities, IEEE Transactions on
  Information Theory 62~(11) (2016) 5973--6006.

\bibitem{Yamano2021}
T.~Yamano, Skewed {J}ensen--{F}isher divergence and its bounds, Foundations
  1~(2) (2021) 256--264.

\bibitem{Zozor2021}
S.~Zozor, J.-F. Bercher, $\phi$-informational measures: Some results and
  interrelations, Entropy 23~(7) (2021) 911.

\bibitem{Wu2025}
H.~Wu, L.~Yu, Entropic isoperimetric and {C}ram{\'e}r--{R}ao inequalities for
  {R}{\'e}nyi--{F}isher information, IEEE Transactions on Information Theory
  (2025).

\bibitem{Stam1959}
A.~J. Stam, Some inequalities satisfied by the quantities of information of
  fisher and shannon, Information and Control 2~(2) (1959) 101--112.
\newblock \href {https://doi.org/10.1016/S0019-9958(59)90348-1}
  {\path{doi:10.1016/S0019-9958(59)90348-1}}.

\bibitem{Shannon1948}
C.~E. Shannon, A mathematical theory of communication, The Bell System
  Technical Journal 27~(3) (1948) 379--423.

\bibitem{Wehrl1978}
A.~Wehrl, General properties of entropy, Reviews of Modern Physics 50~(2)
  (1978) 221--260.
\newblock \href {https://doi.org/10.1103/RevModPhys.50.221}
  {\path{doi:10.1103/RevModPhys.50.221}}.

\bibitem{Rao1945}
C.~R. Rao, Information and the accuracy attainable in the estimation of
  statistical parameters, Bull. Calcutta Math. Soc 37~(3) (1945) 81--91.

\bibitem{Cramer1946}
H.~Cram{\'e}r, Mathematical {M}ethods of {S}tatistics, Vol.~9 of Princeton
  Mathematical Series, Princeton University Press, Princeton, NJ, 1946.

\bibitem{Lutwak2004}
E.~Lutwak, D.~Yang, G.~Zhang, Moment-entropy inequalities, The Annals of
  Probability 32~(1B) (2004) 757--774.

\bibitem{Lutwak2005}
E.~Lutwak, D.~Yang, G.~Zhang, Cram\'er--{R}ao and moment-entropy inequalities
  for \text{R}\'enyi entropy and generalized {F}isher information, IEEE
  Transactions on Information Theory 51~(2) (2005) 473--478.

\bibitem{Bercher2012}
J.-F. Bercher, On generalized {C}ram{\'e}r--{R}ao inequalities, generalized
  {F}isher information and characterizations of generalized {$q-$}{G}aussian
  distributions, Journal of Physics A: Mathematical and Theoretical 45~(25)
  (2012) 255303.

\bibitem{Zozor2017}
S.~Zozor, D.~Puertas-Centeno, J.~S. Dehesa, On generalized {S}tam inequalities
  and {F}isher--\text{R}{\'e}nyi complexity measures, Entropy 19~(9) (2017)
  493.

\bibitem{Puertas2025}
D.~Puertas-Centeno, S.~Zozor,
  \href{https://dx.doi.org/10.1088/1751-8121/adc960}{Some informational
  inequalities involving generalized trigonometric functions and a new class of
  generalized moments}, Journal of Physics A: Mathematical and Theoretical
  58~(16) (2025) 165002.
\newblock \href {https://doi.org/10.1088/1751-8121/adc960}
  {\path{doi:10.1088/1751-8121/adc960}}.
\newline\urlprefix\url{https://dx.doi.org/10.1088/1751-8121/adc960}

\bibitem{IP2025}
R.~G. Iagar, D.~Puertas-Centeno, A new pair of transformations and applications
  to generalized informational inequalities and {H}ausdorff moment problem,
  Communications in Nonlinear Science and Numerical Simulation 151 (2025)
  109091.
\newblock \href {https://doi.org/https://doi.org/10.1016/j.cnsns.2025.109091}
  {\path{doi:https://doi.org/10.1016/j.cnsns.2025.109091}}.

\bibitem{IP2025(b)}
R.~G. Iagar, D.~Puertas-Centeno, Through and beyond moments, entropies and
  {F}isher information measures: new informational functionals and
  inequalities, Physica D: Nonlinear Phenomena 483 (2025) 134928.

\bibitem{IP2026b}
R.~G. Iagar, D.~Puertas-Centeno,
  \href{https://www.sciencedirect.com/science/article/pii/S1007570426004405}{Generalized
  informational functionals and new monotone measures of statistical
  complexity}, Communications in Nonlinear Science and Numerical Simulation 161
  (2026) 110081.
\newblock \href {https://doi.org/https://doi.org/10.1016/j.cnsns.2026.110081}
  {\path{doi:https://doi.org/10.1016/j.cnsns.2026.110081}}.
\newline\urlprefix\url{https://www.sciencedirect.com/science/article/pii/S1007570426004405}

\bibitem{IPT2025}
R.~G. Iagar, D.~Puertas-Centeno, E.~V. Toranzo, Sharp informational
  inequalities involving {K}ullback-{L}eibler and {R}ényi divergences and a
  family of scaling-invariant relative {F}isher measures, arXiv preprint
  arXiv:2507.17408 (2025).

\bibitem{IP2026}
R.~G. Iagar, D.~Puertas-Centeno, Entropies, cross-entropies and {R}\'enyi
  divergence: sharp three-term inequalities for probability density functions,
  arXiv preprint arXiv:2603.07995 (2026).

\bibitem{Bercher2012a}
J.-F. Bercher, On a {$(\beta,q)-$}generalized {F}isher information and
  inequalities involving {$q-$}{G}aussian distributions, Journal of
  Mathematical Physics 53~(6) (2012).

\bibitem{Drabek1999}
P.~Dr{\'a}bek, R.~Man{\'a}sevich, On the closed solution to some nonhomogeneous
  eigenvalue problems with {$p-$}{L}aplacian, Differential Integral Equations
  12~(6) (1999) 773--788.

\bibitem{McDonald_Xu1995}
J.~B. McDonald, Y.~J. Xu, A generalization of the beta distribution with
  applications, Journal of Econometrics 66~(1-2) (1995) 133--152.

\end{thebibliography}

\end{document}